\documentclass[aps, pra, twocolumn, superscriptaddress, nofootinbib, showpacs]{revtex4-1}
\usepackage{amsfonts,amssymb,amsmath}
\usepackage{amsthm}
\usepackage{graphics,graphicx,epsfig}
\usepackage[dvipsnames]{xcolor}
\usepackage[normalem]{ulem}
\usepackage{comment}

\newcommand{\ket}[1]{\mbox{$ | #1 \rangle $}}
\newcommand{\bra}[1]{\mbox{$ \langle #1 | $}}
\newcommand{\braket}[2]{\mbox{$ \langle #1 | #2 \rangle $}}
\newcommand{\cH}{{\cal H}}
\newcommand{\tr}{\mathrm{tr}}

\newtheorem{proposition}{Proposition}

\newtheorem{conjecture}{Conjecture}

\newcommand{\mean}[1]{\langle #1 \rangle}
\newcommand{\Trnorm}[1]{{\parallel\! #1 \!\parallel_{\rm tr}}}

\newcommand{\one}{\openone}

\newcommand{\Tr}[1]{\rm {Tr}}

\renewcommand{\vec}[1]{\mathbf{#1}}

\newcommand{\beq}{\begin{equation}}
\newcommand{\eeq}{\end{equation}}

\begin{document}
\title{Enhanced entanglement criterion via symmetric informationally complete measurements}
\author{Jiangwei Shang}
\email{jiangwei.shang@bit.edu.cn}
\affiliation{Beijing Key Laboratory of Nanophotonics and Ultrafine Optoelectronic Systems, School of Physics, Beijing Institute of Technology, Beijing 100081, China}
\affiliation{Naturwissenschaftlich-Technische Fakult{\"a}t, Universit{\"a}t Siegen, Walter-Flex-Stra{\ss}e 3, 57068 Siegen, Germany}
\author{Ali Asadian}
\email{ali.asadian668@gmail.com}
\affiliation{Naturwissenschaftlich-Technische Fakult{\"a}t, Universit{\"a}t Siegen, Walter-Flex-Stra{\ss}e 3, 57068 Siegen, Germany}
\affiliation{Department of Physics, Institute for Advanced Studies in Basic Sciences (IASBS), Gava Zang, Zanjan 45137-66731, Iran}
\affiliation{Vienna Center for Quantum Science and Technology, Atominstitut, TU Wien, 1040 Vienna, Austria}
\author{Huangjun Zhu}
\email{zhuhuangjun@fudan.edu.cn}
\affiliation{Department of Physics and Center for Field Theory and Particle Physics, Fudan University, Shanghai 200433, China}
\affiliation{State Key Laboratory of Surface Physics, Fudan University, Shanghai 200433, China}
\affiliation{Institute for Nanoelectronic Devices and Quantum Computing, Fudan University, Shanghai 200433, China}
\affiliation{Collaborative Innovation Center of Advanced Microstructures, Nanjing 210093, China}
\affiliation{Institute for Theoretical Physics, University of Cologne,  Cologne 50937, Germany}
\author{Otfried G\"{u}hne}
\email{otfried.guehne@uni-siegen.de}
\affiliation{Naturwissenschaftlich-Technische Fakult{\"a}t, Universit{\"a}t Siegen, Walter-Flex-Stra{\ss}e 3, 57068 Siegen, Germany}

\date{\today}
%

\begin{abstract}
We show that a special type of measurements, called  symmetric informationally complete positive operator-valued measures (SIC POVMs), provide a stronger entanglement detection criterion than the computable cross-norm or realignment criterion based on local orthogonal observables. As an illustration, we demonstrate the enhanced entanglement detection power in simple systems of qubit and qutrit pairs.
This observation highlights the significance  of SIC POVMs for entanglement detection.
\end{abstract}


\maketitle

\section{Introduction}
Entanglement \cite{ent1,ent2} is one of the most distinctive features of quantum theory as compared to classical theory, which is also considered to be a useful resource for tasks like quantum communication, quantum cryptography, and quantum metrology.
Thus, developing simple and efficient criteria for the detection of entanglement in quantum states is indeed pivotal.
However, it has long been proven that entanglement detection is an NP-hard problem as the system size increases \cite{Gurvits2003, Gharibian2010}.

A pure state $\ket{\psi}$ of two particles is called separable, if it is a product
state $\ket{\psi}={\ket{u}\otimes \ket{v}}$, otherwise it is entangled. More generally, a bipartite
mixed state is separable if it can be written as a convex combination
of pure product states,
\begin{equation}\label{eq:convexsum}
\rho = \sum_i p_i \ket{u_i}\bra{u_i} \otimes \ket{v_i}\bra{v_i}\,,
\end{equation}
where the $p_i$s form a probability distribution, so they are positive and sum
up to one. A state that cannot be written in the above form is called entangled.

Many criteria on entanglement detection have been developed. Most of them, however, provide only sufficient conditions for the detection. A well-known example is the positive partial transpose (PPT) criterion \cite{ppt, ppt2}, which is necessary and sufficient for qubit-qubit and qubit-qutrit systems only.
For higher dimensions, there exist the so-called bound entangled states which are PPT and nondistillable.
It is still an open question whether the PPT criterion completely characterizes bound entanglement, namely, whether all bound entangled states are PPT.

Another popular criterion is the computable cross-norm or realignment criterion \cite{ccnr1, ccnr2, ccnr3}, often acronymed as CCNR.
Notably, the CCNR criterion is able to detect the entanglement of many states where the PPT criterion fails.
On the other hand, there also exist some states which are detected by the PPT criterion, but cannot be detected by CCNR \cite{pra78.052319}.
Therefore, we should not assess one as either stronger or weaker than the other criterion, but rather complementing each other.
There are also nonlinear extensions of the CCNR criterion using, for instance, the local uncertainty relations or covariance matrices \cite{prl99.130504, pra78.052319, pra81.032333, pra82.032306}, and other feasible methods \cite{pra77.060301}.
Nevertheless, here we will not consider such extensions.

Empirically, the CCNR criterion can be evaluated by measuring the correlations between local orthogonal observables of two parties. Here, we propose an analogous yet more efficient entanglement detection criterion. Instead of using a set of local orthogonal observables, we use a single generalized measurement,  known as  the symmetric informationally complete positive operator-valued measure (SIC POVM), for each party. Our criterion is determined by the correlations between POVM elements of the two SIC POVMs for the two parties
and shares the simplicity of the CCNR criterion. In addition, it has at least two notable advantages, which are connected to the properties of SIC POVMs. First, this criterion can detect many entangled states (including bound entangled states) that cannot be detected by the CCNR criterion. Second, all correlations featured in the criterion can be measured in one go instead of one by one as in the CCNR criterion.  These advantages are expected to have both theoretical and experimental interests.

Incidentally, Refs.~\cite{qip14.2281, qip15.5119} used a different approach, namely, the sum of the correlation entries, to derive  separability criteria with SIC POVMs.
Very recently, Bae \emph{et al.} \cite{arXiv:1803.02708} studied entanglement detection via quantum 2-designs, which include  SIC POVMs as a special example; however, their entanglement criteria are not so closely related to the properties of 2-designs except for the tomographic completeness.  Ref.~\cite{arXiv:1801.07927} investigated the entanglement properties of multipartite systems with tight informationally complete measurements, including SIC POVMs.
In addition,  Ref.~\cite{qip17.111}  considered a nonlinear entanglement criterion based on SIC POVMs, which turns out to be  equivalent to the criteria using observables in Refs.~\cite{pra78.052319, pra77.060301}.

This paper is organized as follows. We will first recall the CCNR criterion in Sec.~\ref{sec:CCNR} and the basic properties of  SIC POVMs in Sec.~\ref{sec:SIC}. In Sec.~\ref{sec:EntSIC}, we derive the entanglement criterion based on SIC POVMs. Section~\ref{sec:Ex} demonstrates the superiority of the criterion with various examples, then we close with a few remarks.

\section{The CCNR criterion}\label{sec:CCNR}
The CCNR criterion can be formulated in different forms.
One approach makes use of the Schmidt decomposition of the quantum state in operator space. According to the Schmidt theorem, any bipartite density matrix $\rho$ on ${\cH=\cH^{A}\otimes\cH^{B}}$ with dimension $d=d_A\times d_B$ (assuming $d_A\leq d_B$) can be written as
\begin{equation}\label{eq:schmidtRho}
\rho=\sum_{k=1}^{d_A^2}\lambda_k G_k^{A}\otimes G_k^{B}\,,\quad \lambda_k=\mean{G_k^{A}\otimes G_k^{B}}.
\end{equation}
Here ${\lambda_k\ge0}$ are the Schmidt coefficients,
$\{G_k^{A}\}$ is an orthonormal basis
of Hermitian operators on $\cH^A$, and $\{G_k^{B}\}$ is a set of  orthonormal
Hermitian operators on $\cH^B$, that is, $\tr(G_k^AG_l^A)=\tr(G_k^BG_l^B)=\delta_{kl}$. In terms of the Schmidt decomposition, the CCNR criterion can be stated as  follows.
\begin{proposition}[CCNR]\label{ccnr}
If a state $\rho$ is separable and has Schmidt decomposition as in Eq.~\eqref{eq:schmidtRho}, then
\begin{equation}
\sum_k\lambda_k\le 1
\end{equation}
has to hold; otherwise, it is entangled.
\end{proposition}
\begin{proof}
See, for instance, the proof in Refs.~\cite{ccnr1,ccnr2}.
\end{proof}
Incidentally, the CCNR criterion is closely related to the linear entanglement witness $W=\openone-\sum_kG_k^{A}\otimes G_k^{B}$; see Ref.~\cite{pra69.022312} for another form.

Besides the Schmidt form, the state $\rho$ can also be decomposed with arbitrary local orthonormal operator bases, i.e., $\rho=\sum_{ij}c_{ij} A_i\otimes B_j$, where $\{A_i\}$ ($\{B_j\}$) forms an orthonormal basis for the space of linear operators  on $\cH^A$ ($\cH^B$). It turns out the trace norm of the correlation matrix $[C]_{ij}=\langle A_i\otimes B_j\rangle=c_{ij}$ is equal to the sum of the Schmidt coefficients,
\beq
\label{eq:Ctr}
\Trnorm{C}=\sum_k\lambda_k\,,
\eeq
where $\Trnorm{C}\equiv\tr(\sqrt{CC^\dag})$ denotes the trace norm. This equality follows from the fact that the Schmidt coefficients $\lambda_k$ happen to be the singular values of $C$.

In terms of the correlation matrix, the CCNR criterion can be formulated as follows: the trace norm of the correlation matrix of any separable state is upper bounded by one, that is,
\begin{equation}
\Trnorm{C}\leq 1.
\end{equation}
This conclusion can be proved directly as follows. For any given  product state $\rho=\rho_A\otimes\rho_B$, the matrix elements of the correlation matrix have the form $[C]_{ij}=\langle A_i\rangle \langle B_j\rangle= a_ib_j$, where $a_i=\langle A_i\rangle$ and $b_j=\langle B_j\rangle$. So the correlation matrix can be expressed as an outer product
 $C\equiv|\vec a)(\vec b|$, where  the curly ket $|\vec a)\equiv (a_1\,\, a_2\dots)^T$ denotes the column vector composed of $a_i$s, while the curly bra $(\vec b|\equiv (b_1\,\, b_2\dots)$ denotes the row vector composed of $b_i$s. In addition, we have
\begin{align}
\label{eq:Cnorm}
\Trnorm{C}&=(\vec a|\vec a)^{1/2}(\vec b|\vec b)^{1/2} =\sqrt{\sum_i a^2_i \sum_j b^2_j}\nonumber  \\
&=\sqrt{\tr(\rho_A^2) \tr(\rho_B^2)}\leq 1.
\end{align}
By employing the convexity property of the trace norm under mixing, we conclude that $\Trnorm{C}\leq1$ for any separable state.

\section{SIC POVM}\label{sec:SIC}
In a $d$-dimensional Hilbert space, a SIC POVM  $\mathcal M_{\rm s}=\{\Pi_k\}_{k=1}^{d^2}$ is composed of $d^2$ outcomes $\Pi_k=\ket{\psi_k}\bra{\psi_k}/d$ which are subnormalized rank-1 projectors onto pure states, with equal pairwise fidelity, that is,
\begin{equation}\label{eq:SIC}
|\braket{\psi_i}{\psi_j}|^2=\frac{d \delta_{ij}+1}{d+1}\,,\quad i,j=1,2,\dots,d^2\,.
\end{equation}
It is not difficult to verify the completeness condition, that is,  ${\sum_{k=1}^{d^2}\Pi_k=\one}$.
It has been conjectured that SIC POVMs exist in all finite dimensions \cite{Zauner}, although a general proof is still missing.
So far analytical solutions have been found in dimensions $d=2-24, 28, 30, 31, 35, 37, 39, 43, 48, 124$, and numerical solutions with high precision have been found up to dimension $d=151$; see  Ref.~\cite{SICreview} for a recent review.
Experimentally, SIC POVMs in low dimensions have been realized in various quantum systems \cite{pra74.022309, pra83.051801, prx5.041006}.

Given a SIC POVM $\mathcal M_{\rm s}=\{\Pi_k\}_{k=1}^{d^2}$ and a quantum state $\rho$, the probability of obtaining outcome $k$ is given by the Born rule,
${p_k=\langle\Pi_k\rangle=\tr(\rho\Pi_k)}$. Conversely,  the quantum state $\rho$ can be reconstructed from these probabilities as follows,
\begin{eqnarray}\label{eq:rhoSIC}
\rho&=&\sum_{k=1}^{d^2}\bigl[d(d+1)p_k-1\bigr]\Pi_k\nonumber\\
&=&d(d+1)\sum_{k=1}^{d^2}p_k\Pi_k-\openone\,,
\end{eqnarray}
see Fig.~\ref{fig:SIC-sketch}(a). Calculation shows that
\begin{equation}\label{eq:pureSIC1}
\sum_{k=1}^{d^2}p_k^2=\frac{1+\tr(\rho^2)}{d(d+1)}\leq \frac{2}{d(d+1)}\,,
\end{equation}
where the upper bound is saturated iff $\rho$ is pure.

\begin{figure}[t]
\includegraphics[width=.82\columnwidth]{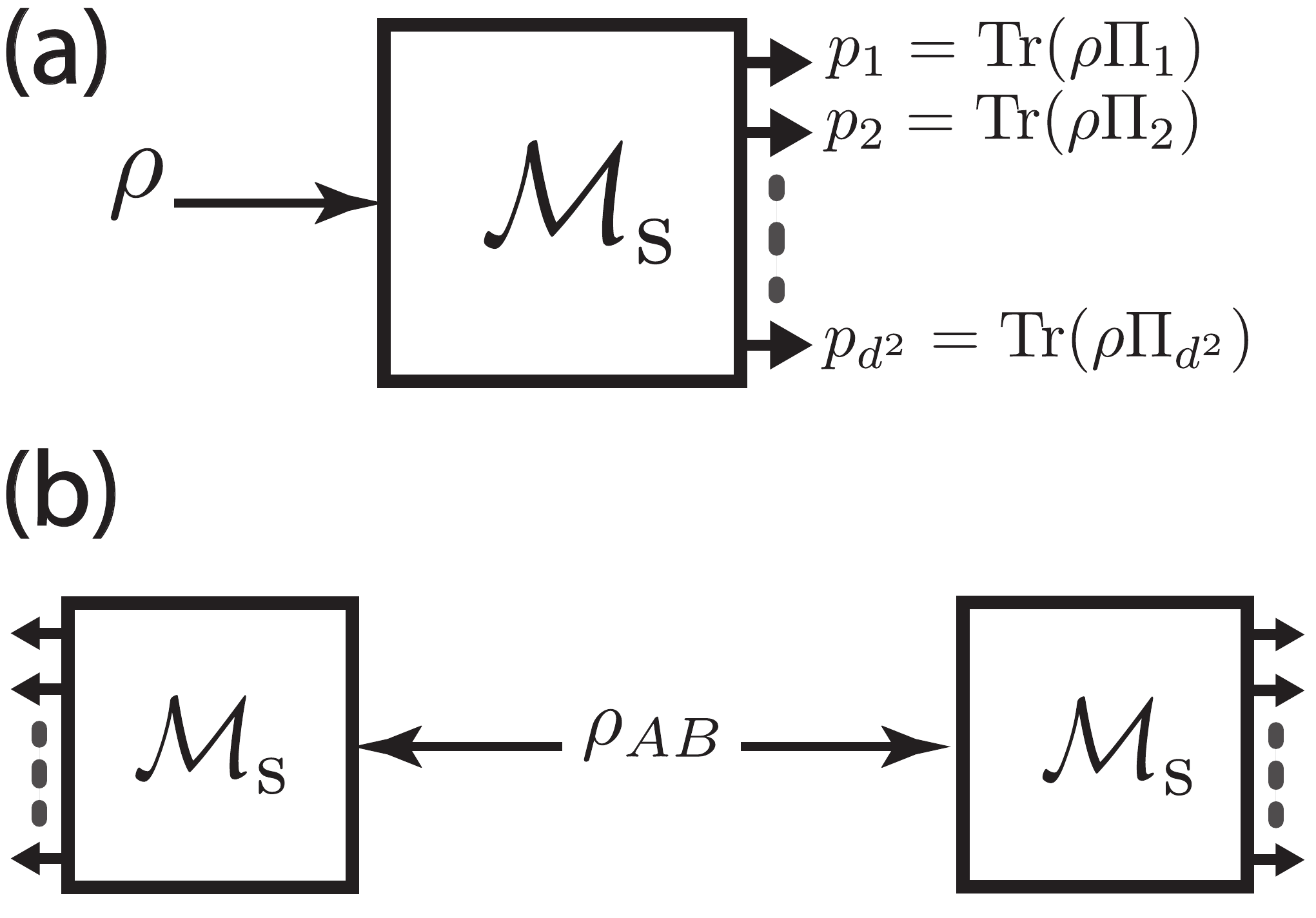}
\caption{\label{fig:SIC-sketch}
(a) Tomographically complete single measurement setting $\mathcal M_{\rm s}$ with $d^2$ outcomes.
(b) Employing this type of measurement setting for each party of a bipartite system, we have a new paradigm for entanglement detection. Unlike that for CCNR, only a single measurement setting for each party is involved, which may help to reduce the experimental complexity.
}
\end{figure}

For the convenience of later applications, we need to renormalize the elements in the SIC POVM,
\begin{equation}
E_k\equiv\sqrt{\frac{d(d+1)}{2}}\,\Pi_k=\sqrt{\frac{d+1}{2d}}\ket{\psi_k}\bra{\psi_k}\,.
\end{equation}
Let
\begin{equation}\label{eq:ek}
e_k=\tr(\rho E_k)=\sqrt{\frac{d(d+1)}{2}}\,p_k\,;
\end{equation}
then the constraint  in Eq.~\eqref{eq:pureSIC1} can be formulated as follows,
\begin{equation}\label{eq:pureSIC}
(\vec e|\vec e)=\sum_{k=1}^{d^2}e_k^2=\frac{1+\tr(\rho^2)}{2} \leq 1\,.
\end{equation}
Let $\{G_k\}$ be an arbitrary Hermitian operator basis and $a_k=\tr(\rho G_k)$. Then $(\vec a|\vec a)=\tr(\rho^2)$, so we have
\beq\label{eq:ea}
(\vec e|\vec e)=\frac{1+(\vec a|\vec a)}{2}\,.
\eeq
This equality is instructive to understanding the relation between our entanglement criterion introduced in the next section and the CCNR criterion.

\section{Entanglement detection via SIC POVMs}\label{sec:EntSIC}
Consider a bipartite state $\rho$ acting on the Hilbert space ${\cH=\cH^{A}\otimes\cH^{B}}$ with dimension ${d=d_A\times d_B}$, and denote by $\{E_k^{A}\}_{k=1}^{d_A^2}$ and $\{E_k^{B}\}_{k=1}^{d_B^2}$ the normalized SIC POVMs for the two respective subsystems; see the schematic setup in Fig.~\ref{fig:SIC-sketch}(b). The linear correlations between $E^{A}$ and $E^{B}$ read
\begin{equation}\label{eq:CorrelationSIC}
[\mathcal P_{\rm s}]_{ij}=\langle E_i^{A}\otimes E_j^{B}\rangle,
\end{equation}
from which we can construct a simple, but useful entanglement criterion via SIC POVMs (ESIC).
\begin{proposition}[ESIC]\label{entSIC}
If a state $\rho$ is separable, then
\begin{equation}
\Trnorm{\mathcal P_{\rm s}}\le 1
\end{equation}
has to hold; otherwise, it is entangled.
\end{proposition}
\begin{proof}
For a product state $\rho=\rho_A\otimes\rho_B$, we have
\begin{equation}
\mathcal P_{\rm s}=
\left(\begin{array}{c}
  e_{A,1} \\
  e_{A,2} \\
  \vdots
\end{array}\right)
\left(\begin{array}{ccc}
  e_{B,1} & e_{B,2} & \cdots
\end{array}\right)\equiv  |\vec e_A)(\vec  e_B|\,,
\end{equation}
where $e_{A,j}=\tr(\rho E_j^A)$ for $j=1,2, \ldots, d_A^2$ and $e_{B,k}=\tr(\rho E_k^B)$ for $k=1,2, \ldots, d_B^2$; cf. Eq.~\eqref{eq:ek}.
Then,
\begin{eqnarray}\label{eq:Pnorm}
\Trnorm{\mathcal P_{\rm s}}&=&(\vec e_A|\vec e_A)^{1/2}(\vec e_B|\vec e_B)^{1/2}\nonumber\\
&=&\sqrt{\sum_j e_{A,j}^2}\sqrt{\sum_k e_{B,k}^2}
\leq 1
\end{eqnarray}
according  to Eq.~\eqref{eq:pureSIC}.
Again, by employing the convexity property of the trace norm under mixing, we have $\Trnorm{\mathcal P_{\rm s}}\leq1$ for separable states in general.
\end{proof}

\begin{proposition}
The ESIC criterion is independent of the specific SIC POVMs for systems A and B. In other words, any pair of SIC POVMs leads to the same criterion.
\end{proposition}
\begin{proof}
Let  $\{E_k^{A}\}_{k=1}^{d_A^2}$ and  $\{\tilde{E}_k^{A}\}_{k=1}^{d_A^2}$ be two arbitrary normalized SIC POVMs for system A, while $\{E_k^{B}\}_{k=1}^{d_B^2}$ and $\{\tilde{E}_k^{B}\}_{k=1}^{d_B^2}$ are two arbitrary  normalized SIC POVMs for system B. Let ${\mathcal P}_{\rm s}$ be the correlation matrix as defined in Eq.~\eqref{eq:CorrelationSIC} and
\begin{equation}
[\tilde{\mathcal P}_{\rm s}]_{ij}=\langle \tilde{E}_i^{A}\otimes \tilde{E}_j^{B}\rangle\,.
\end{equation}
To prove the proposition it suffices to prove that
${\Trnorm{\tilde{\mathcal P}_{\rm s}}=\Trnorm{\mathcal P_{\rm s}}}$. To this end, note that $\tilde{E}_i^{A}$ can be expressed as a linear combination of $E_j^{A}$, that is, ${\tilde{E}_i^{A}=\sum_j O^A_{ij} E_j^{A}}$, where $O^A_{ij}$ are uniquely determined by  $\{E_k^{A}\}_{k=1}^{d_A^2}$ and  $\{\tilde{E}_k^{A}\}_{k=1}^{d_A^2}$, and form an orthogonal matrix. Similarly, $\tilde{E}_i^{B}=\sum_j O^B_{ij} E_j^{B}$ with $ O^B_{ij}$ forming an orthogonal matrix. Therefore,
\begin{equation}
[\tilde{\mathcal P}_{\rm s}]_{ij}= \sum_{k,l} O^A_{ik} O^B_{jl}\langle E_k^{A}\otimes E_l^{B}\rangle=\sum_{k,l} O^A_{ik} O^B_{jl}[\mathcal P_{\rm s}]_{kl}\,,
\end{equation}
that is, ${\tilde{\mathcal P}_{\rm s}=O^A\mathcal P_{\rm s}(O^B)^T}$. Since both $O^A$ and $O^B$ are orthogonal matrices, it follows that ${\Trnorm{\tilde{\mathcal P}_{\rm s}}=\Trnorm{\mathcal P_{\rm s}}}$, so the entanglement  criterion in Proposition \ref{entSIC} does not depend on the specific SIC POVMs for systems A and B.
\end{proof}

It is not easy to analytically compare the ESIC criterion with CCNR.
Nevertheless, through extensive numerical evidence, we find that the ESIC criterion is  stronger than the CCNR criterion. Here we present several  observations and a conjecture.
For a product state, we have
	\beq
	\Trnorm{C}\leq \Trnorm{\mathcal P_{\rm s}}\,.
	\eeq
This inequality follows from Eqs.~\eqref{eq:Cnorm}, \eqref{eq:ea}, and \eqref{eq:Pnorm}. Numerical calculations suggest that this inequality also holds for separable states.

\begin{conjecture}\label{conj}
If ${\Trnorm{C}>1}$, then ${\Trnorm{\mathcal P_{\rm s}}>1}$.
\end{conjecture}
\noindent\emph{Remark}.
The converse of the above relation does not hold in general.
However, for the special case of bipartite pure states, our numerical calculations suggest that
\begin{equation}\label{eq:CP}
\Trnorm{C}-1=2(\Trnorm{\mathcal P_{\rm s}}-1)\,.
\end{equation}
The same is true if the bipartite state is a convex combination of a pure state and the completely mixed state.
So the ESIC criterion and the CCNR criterion become equivalent in this special case.
Incidentally, in this case, they are also equivalent to the nonlinear criteria presented in Refs.~\cite{pra78.052319, pra77.060301}.
Note that Eq.~\eqref{eq:CP} can be verified analytically in the case of two-qubit pure states.

For a bipartite pure state $\rho=\ket{\psi}\bra{\psi}$, the Schmidt decomposition in Eq.~\eqref{eq:schmidtRho} is closely related to the Schmidt decomposition of $\ket{\psi}$. Suppose $\ket{\psi}$ has the Schmidt decomposition
\begin{equation}
\ket{\psi}=\sum_{k=1}^{d_A}\tilde{\lambda}_k\ket{u_k^A}\ket{v_k^B}\,,
\end{equation}
where $\tilde{\lambda}_k$ are the Schmidt coefficients and satisfy $\sum_k\tilde{\lambda}_k^2=1$. Then the Schmidt coefficicents in Eq.~\eqref{eq:schmidtRho} are given by $\tilde{\lambda}_i\tilde{\lambda}_j$ for $i,j=1,2,\ldots, d_A$. Therefore,
\begin{equation}
\Trnorm{C}=\sum_i\sum_j\tilde{\lambda}_i\tilde{\lambda}_j= \Biggl(\sum_j\tilde{\lambda}_j\Biggr)^2  \,,
\end{equation}
In Fig.~\ref{fig:pure_qutrit}, we plot the value of $\Trnorm{C}-1$ [that is, $2(\Trnorm{\mathcal P_{\rm s}}-1)$ according to Eq.~\eqref{eq:CP}] against the two independent squared Schmidt coefficients for two-qutrit pure states.
All states are entangled, except for the ones corresponding to the  minimum value (the three fulcrums) of the surface plot, which are  product states.
Top of the surface (or center of the contours) represents the maximally entangled states. Incidentally, for two-qubit pure states, the corresponding plot  is an arc lying on the plane formed by one of the horizontal axes and the vertical axis.

In the following section, we use various examples to demonstrate that Conjecture~\ref{conj} holds true.

\begin{figure}[t]
\includegraphics[width=.9\columnwidth]{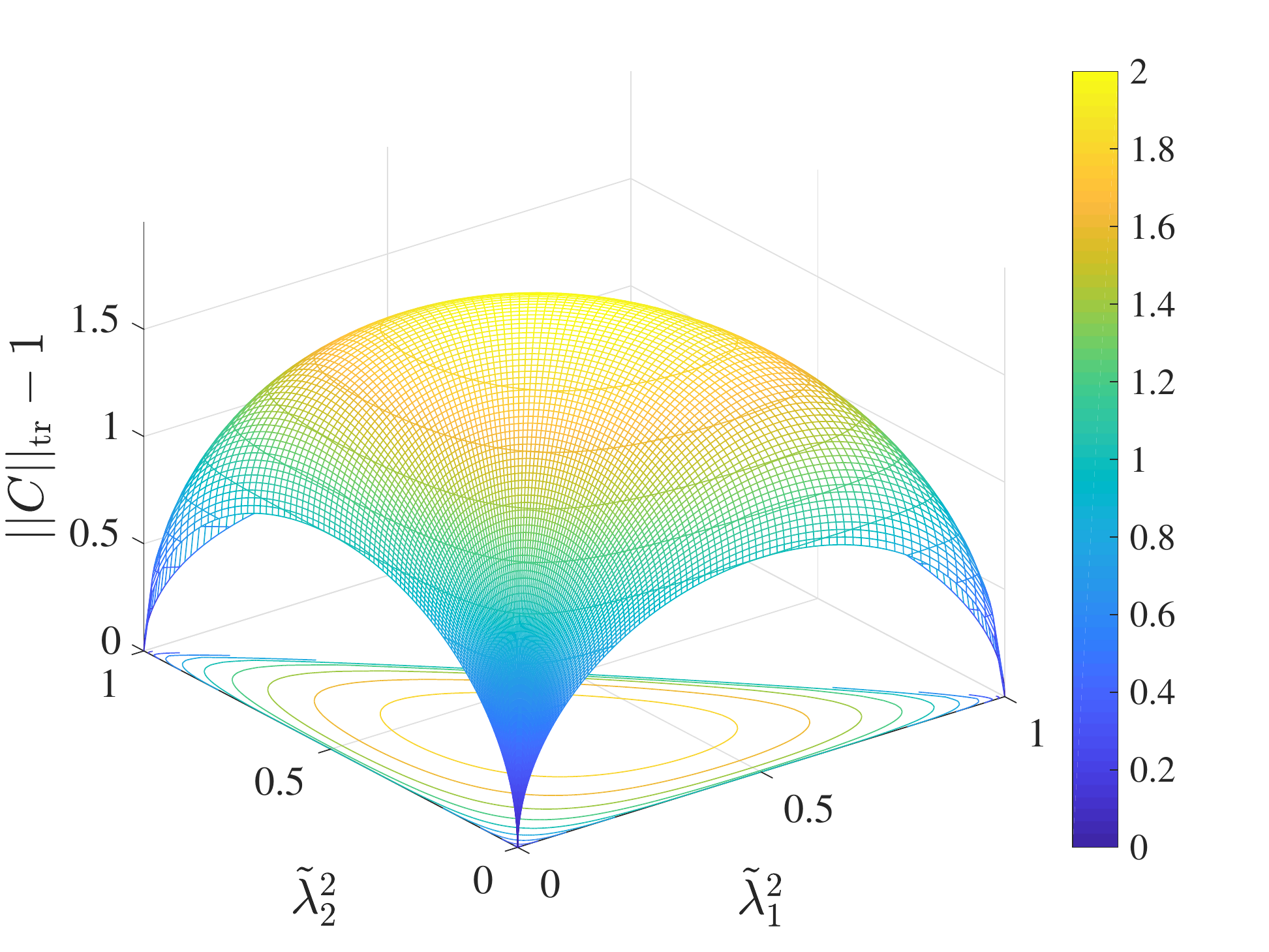}
\caption{\label{fig:pure_qutrit}
The value of $\Trnorm{C}-1$ [or $2(\Trnorm{\mathcal P_{\rm s}}-1)$ according to {Eq.~\eqref{eq:CP}}] with respect to the two independent squared Schmidt coefficients for two-qutrit pure states.
The  minimum (the three fulcrums) corresponds to product states, while the maximum (top of the surface plot, or center of the contours) corresponds to the maximally entangled states.
}
\end{figure}
%

\section{Examples}\label{sec:Ex}
For the first example, let us consider  two-qubit states.
In this case, the PPT criterion \cite{ppt, ppt2} is necessary and sufficient for detecting entanglement.
For a randomly generated two-qubit state $\rho_{\mathrm{2qb}}$, if it is entangled, then we add white noise to it and get the state
\begin{equation}
\rho(q)=q\rho_{\mathrm{2qb}}+(1-q)\frac{\one}{4},\quad 0\leq q\leq 1.
\end{equation}
The threshold value of $q$ at which $\rho(q)$ becomes separable can be determined by the PPT criterion $\Trnorm{\rho^{T_B}}$. Then we compare it with the values determined by the ESIC $\Trnorm{\mathcal P_{\rm s}}$ and CCNR $\Trnorm{C}$ criteria, respectively. The smaller the difference is as compared to the PPT criterion, the better the criterion is. Figure~\ref{fig:2qb} illustrates the results on $10\,000$ entangled two-qubit states which are generated randomly according to the Hilbert-Schmidt measure
\cite{note1}.
As can be seen, the ESIC citerion is better than CCNR in most of the cases.
Notably, all states that are detected by CCNR can also be detected by ESIC.

Moreover, the advantage of ESIC over CCNR is not only tied to the Hilbert-Schmidt measure. To corroborate this point, we have considered random mixed states generated according to various different measures. In particular, we have studied induced measures on mixed states obtained by taking partial trace of the Haar random pure states of  bipartite systems \cite{pra71.032313}. For example, if $|\Psi\rangle$ is a Haar random pure state in dimension $N\times K$, then taking partial trace over the second system yields a random mixed state $\rho=\tr_K(\ket{\Psi}\bra{\Psi})$ on the Hilbert space of dimension $N$. The resulting induced measure will be denoted by $(N,K)$. Note that this measure is equivalent to the Hilbert-Schmidt measure  when $N=K$.
In general, as $K$ increases, 
the  measure gets more biased towards more mixed states. Here we are interested in the two-qubit random mixed states, so $N=2\times 2=4$. Table~\ref{tab:2qb} shows the distinction between ESIC and CCNR for three different choices of $K$, namely $K=3,4,6$. In all three cases, ESIC can detect more entangled states than CCNR. In addition, we have considered random states generated according to the Jeffreys prior  over the probability space \cite{arXiv:1612.05180}.
Again, the ESIC criterion is better than CCNR.

We mentioned early that the CCNR criterion is able to detect certain bound entangled states, where the PPT criterion  fails completely.
In the second example, we consider the family of $3\times3$ bound entangled states introduced by P. Horodecki \cite{pla232.333},
\begin{equation}\label{eq:rhoPH}
\rho_{\mathrm{PH}}^{x}=\frac1{8x+1}
\left(
\begin{array}{ccccccccc}
 x & 0 & 0 & 0 & x & 0 & 0 & 0 & x \\
 0 & x & 0 & 0 & 0 & 0 & 0 & 0 & 0 \\
 0 & 0 & x & 0 & 0 & 0 & 0 & 0 & 0 \\
 0 & 0 & 0 & x & 0 & 0 & 0 & 0 & 0 \\
 x & 0 & 0 & 0 & x & 0 & 0 & 0 & x \\
 0 & 0 & 0 & 0 & 0 & x & 0 & 0 & 0 \\
 0 & 0 & 0 & 0 & 0 & 0 & \frac{1+x}{2} & 0 & \frac{\sqrt{1-x^2}}{2} \\
 0 & 0 & 0 & 0 & 0 & 0 & 0 & x & 0 \\
 x & 0 & 0 & 0 & x & 0 & \frac{\sqrt{1-x^2}}{2} & 0 & \frac{1+x}{2}
\end{array}
\right)\!.
\end{equation}
Although these states cannot be detected by the PPT criterion and are not distillable, they are nevertheless entangled for  ${0<x<1}$. Consider the mixtures of $\rho_{\mathrm{PH}}^{x}$ with the white noise,
\begin{equation}
\rho(x,q)=q\rho_{\mathrm{PH}}^{x}+(1-q)\frac{\one}{9},\quad 0\leq q\leq 1.
\end{equation}
Figure~\ref{fig:PH} illustrates the parameter range for which the state $\rho(x,q)$ is entangled and can be detected by the ESIC (CCNR) criterion. It is clear from the figure that the ESIC criterion can detect strictly more entangled states than the CCNR criterion.

\begin{figure}[t]
\includegraphics[width=.9\columnwidth]{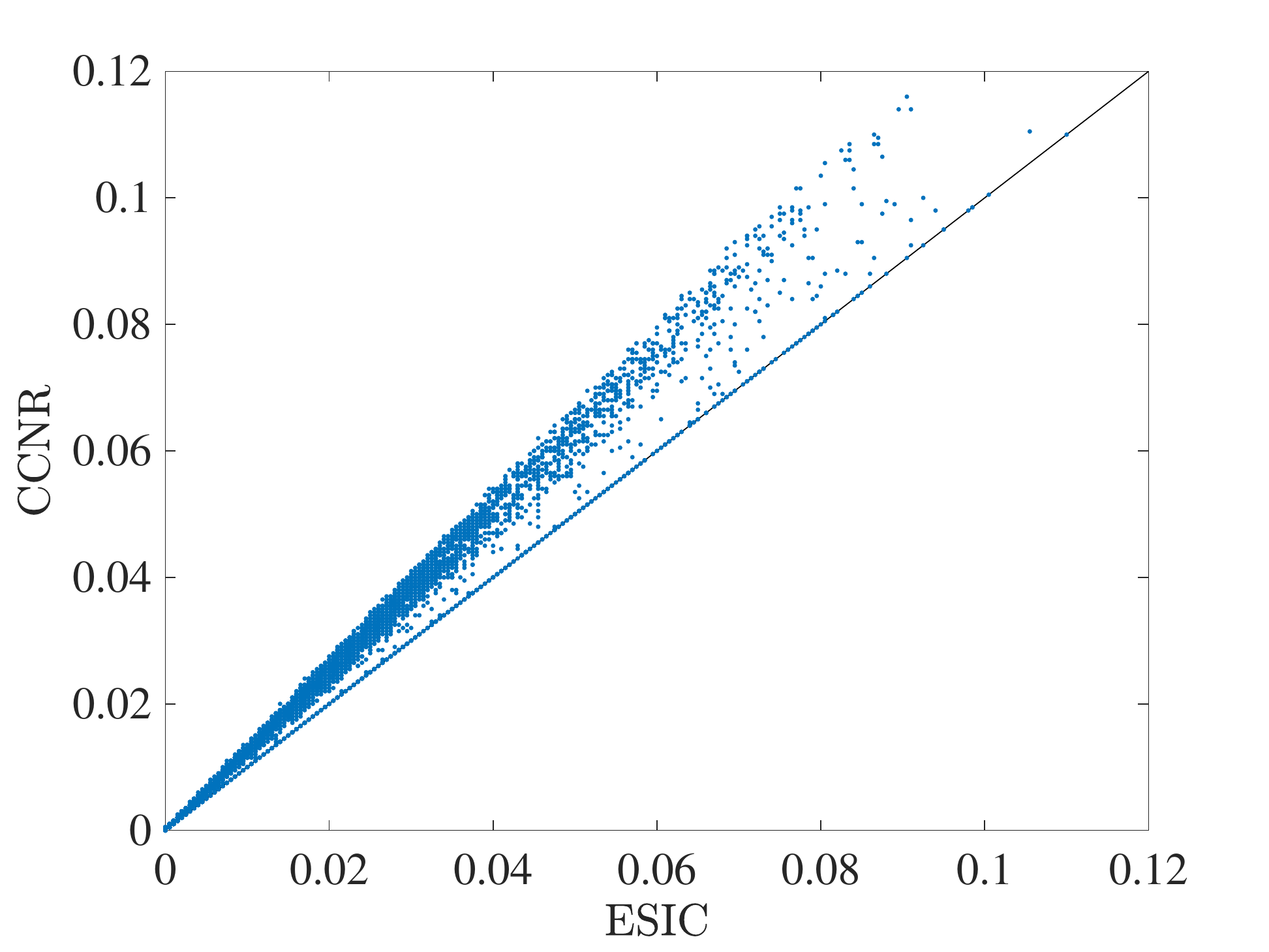}
\caption{\label{fig:2qb}
For a randomly generated two-qubit entangled state, we add white noise to it until it cannot be detected by a particular criterion.
The horizontal (vertical) axis represents the difference in the threshold value between the ESIC (CCNR) criterion and the PPT criterion.
The smaller the difference is, the better the criterion is. 
The plot illustrates the results on $10\,000$ entangled two-qubit states which are generated randomly according to the Hilbert-Schmidt measure \cite{note1},  which demonstrates that the ESIC criterion  is much better than the CCNR  criterion.
}
\end{figure}
\begin{table}[h]
\caption{\label{tab:2qb}
Entanglement detection of  two-qubit states generated  randomly according to various measures. Altogether $50\,000$ states are generated in each case, and the numbers show the fractions of states that are detected.
Here the symbol $(N,K)$ denotes the family of measures induced by the Haar measure on the unitary group $U(NK)$, and in the two-qubit case $N=4$; see Ref.~\cite{pra71.032313} for details.
The symbol ``Jeff'' represents the online readily-made random samples in Ref.~\cite{arXiv:1612.05180} that are generated according to the Jeffreys prior over the probability space.
As we can see, the ESIC criterion is better than CCNR no matter what measure we choose to generate the random samples.}
\begin{tabular}{@{\hspace*{0.2cm}}c@{\hspace*{0.2cm}}|@{\hspace*{0.2cm}}c@{\hspace*{0.2cm}}|@{\hspace*{0.2cm}}c@{\hspace*{0.2cm}}|@{\hspace*{0.2cm}}c@{\hspace*{0.2cm}}|@{\hspace*{0.2cm}}c@{\hspace*{0.2cm}}}
\hline\hline
   & (4, 3) & (4, 4) & (4, 6) & Jeff \\
\hline
PPT & 92.93\% & 75.88\% & 41.35\% & 79.27\% \\
CCNR & 84.60\%  & 65.77\%  & 33.91\% & 69.18\% \\
ESIC & 85.97\% & 67.27\%  & 35.09\% & 70.75\% \\
\hline
\end{tabular}
\end{table}
\begin{figure}[t]
\includegraphics[width=.9\columnwidth]{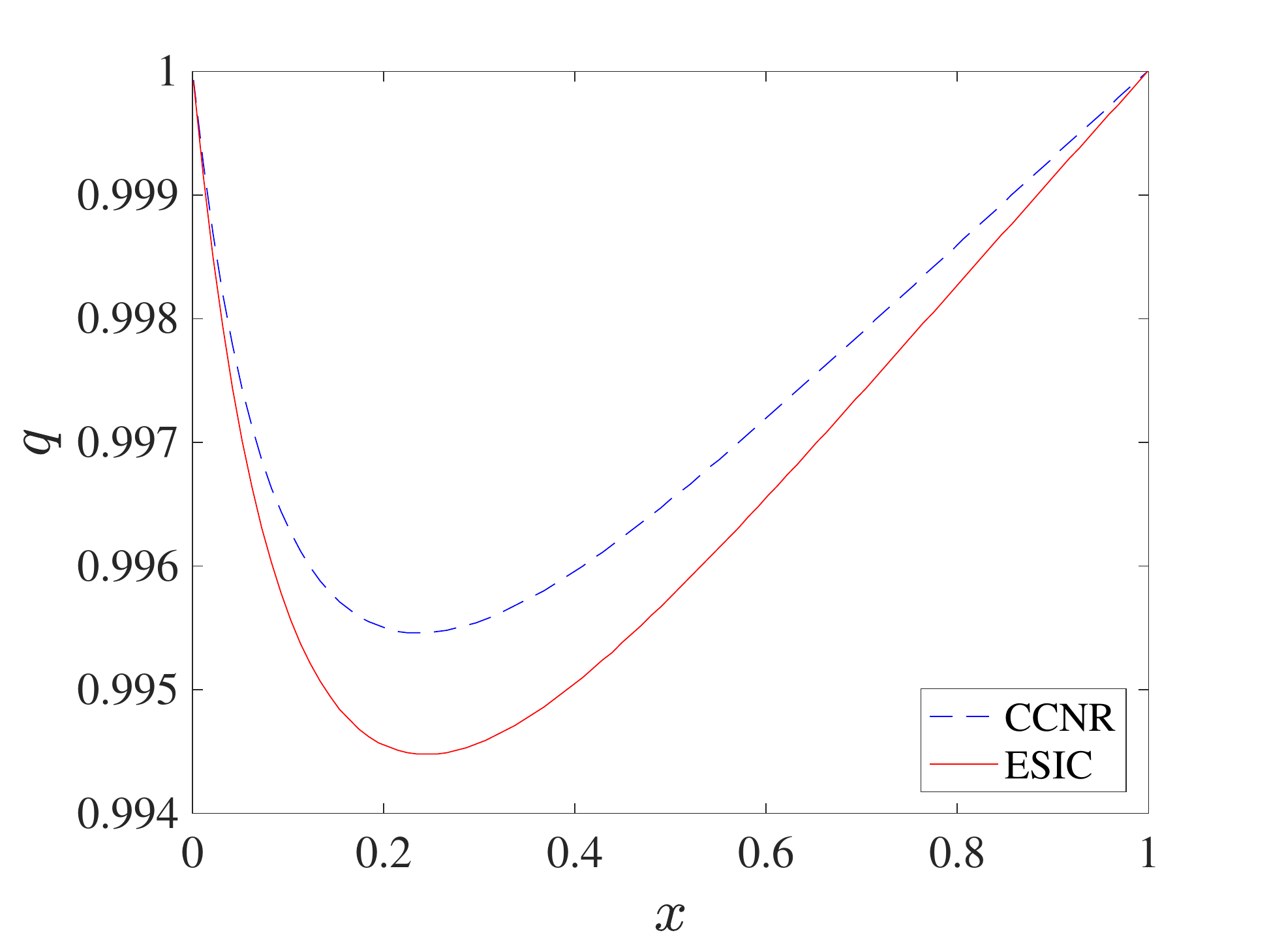}
\caption{\label{fig:PH}
Entanglement detection of $\rho(x,q)$, mixtures of  Horodecki $3\times3$ bound entangled states  with the white noise. States above the red solid (blue dashed) curve
 are detected by the ESIC (CCNR) criterion. The ESIC criterion can detect all states detected by the CCNR criterion and some additional states.
}
\end{figure}

For the next example, we consider the $3\times3$ bound entangled states called \emph{chessboard} states \cite{chessboard}. They are defined as
\begin{equation}
\rho_{\mathrm{cb}}={\cal N}\sum_{j=1}^4\ket{V_j}\bra{V_j}\,,
\end{equation}
with ${\cal N}$ being the normalization constant, where the unnormalized vectors are
\begin{eqnarray}
  \ket{V_1}&=&\ket{m_5,0,m_1m_3/m_6;0,m_6,0;0,0,0}\,,\nonumber\\
  \ket{V_2}&=&\ket{0,m_1,0;m_2,0,m_3;0,0,0}\,,\nonumber\\
  \ket{V_3}&=&\ket{m_6,0,0;0,-m_5,0;m_1m_4/m_5,0,0}\,,\nonumber\\
  \ket{V_4}&=&\ket{0,m_2,0;-m_1,0,0;0,m_4,0}\,.
\end{eqnarray}
These chessboard states are characterized by the six real parameters $m_k$s for $k=1,2,\ldots,6$.
We compare the ESIC criterion with the CCNR criterion on randomly generated chessboard states, where the six parameters are drawn independently from the normal distribution with zero mean value and standard deviation of 2.
Altogether $50\,000$ random states are generated, and the results are shown in Fig.~\ref{fig:chessboard}.
As we can see, the PPT criterion fails completely in this case.
The CCNR criterion detects $18.36\%$  states, while the ESIC criterion is able to detect  $20.37\%$ states, which is roughly $2\%$ more.
Again, all states that are detected by CCNR can also be detected by ESIC.

\begin{figure}[t]
\includegraphics[width=.9\columnwidth]{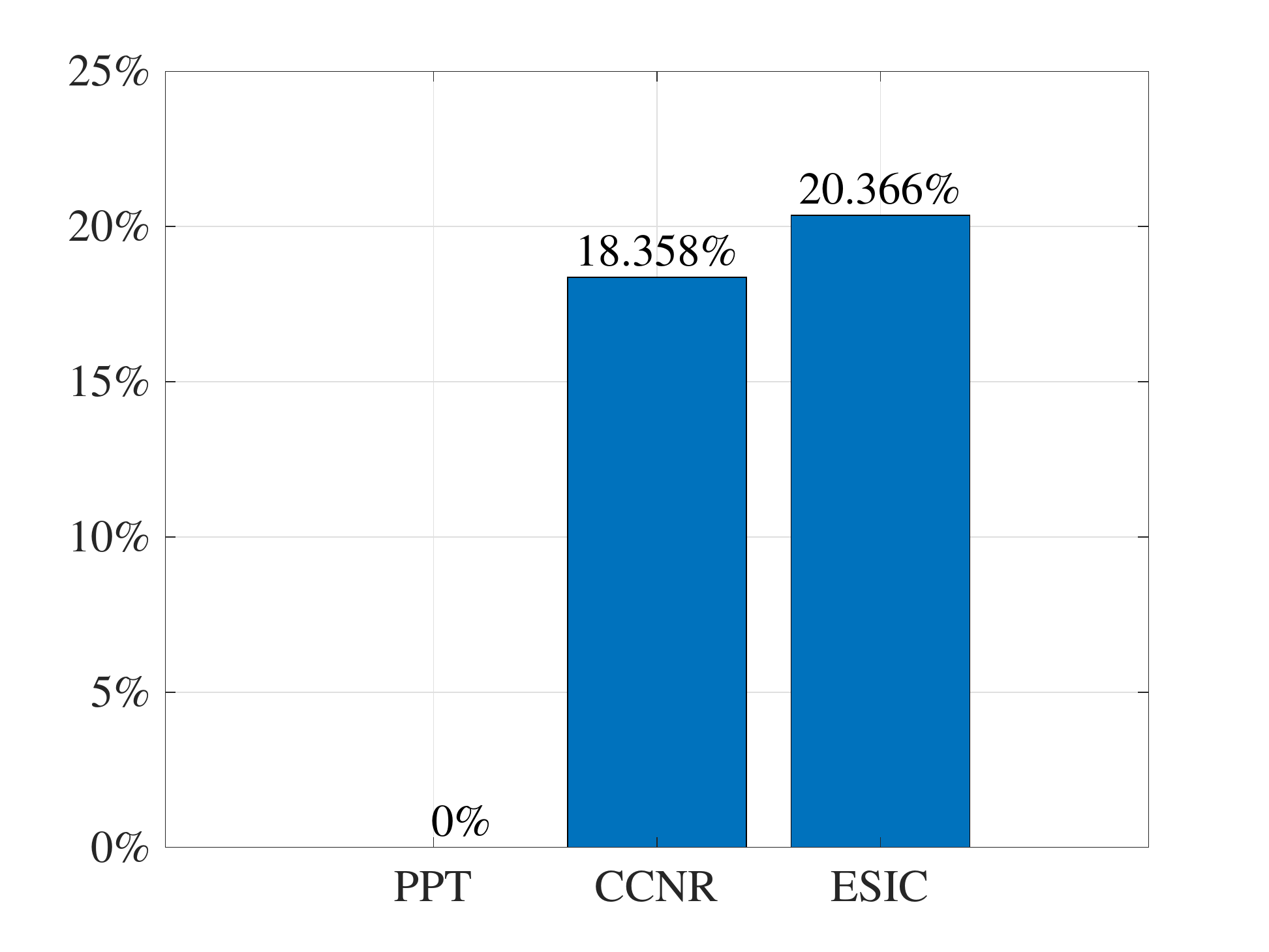}
\caption{\label{fig:chessboard}
Entanglement detection of $3\times3$ chessboard states. Altogether $50\,000$  states are generated randomly, and the plot shows the fractions of states that are detected by  three criteria, respectively. Here the PPT criterion fails completely, the ESIC criterion can detect roughly $2\%$ more states than the  CCNR criterion.
}
\end{figure}
\begin{figure}[t]
\includegraphics[width=.9\columnwidth]{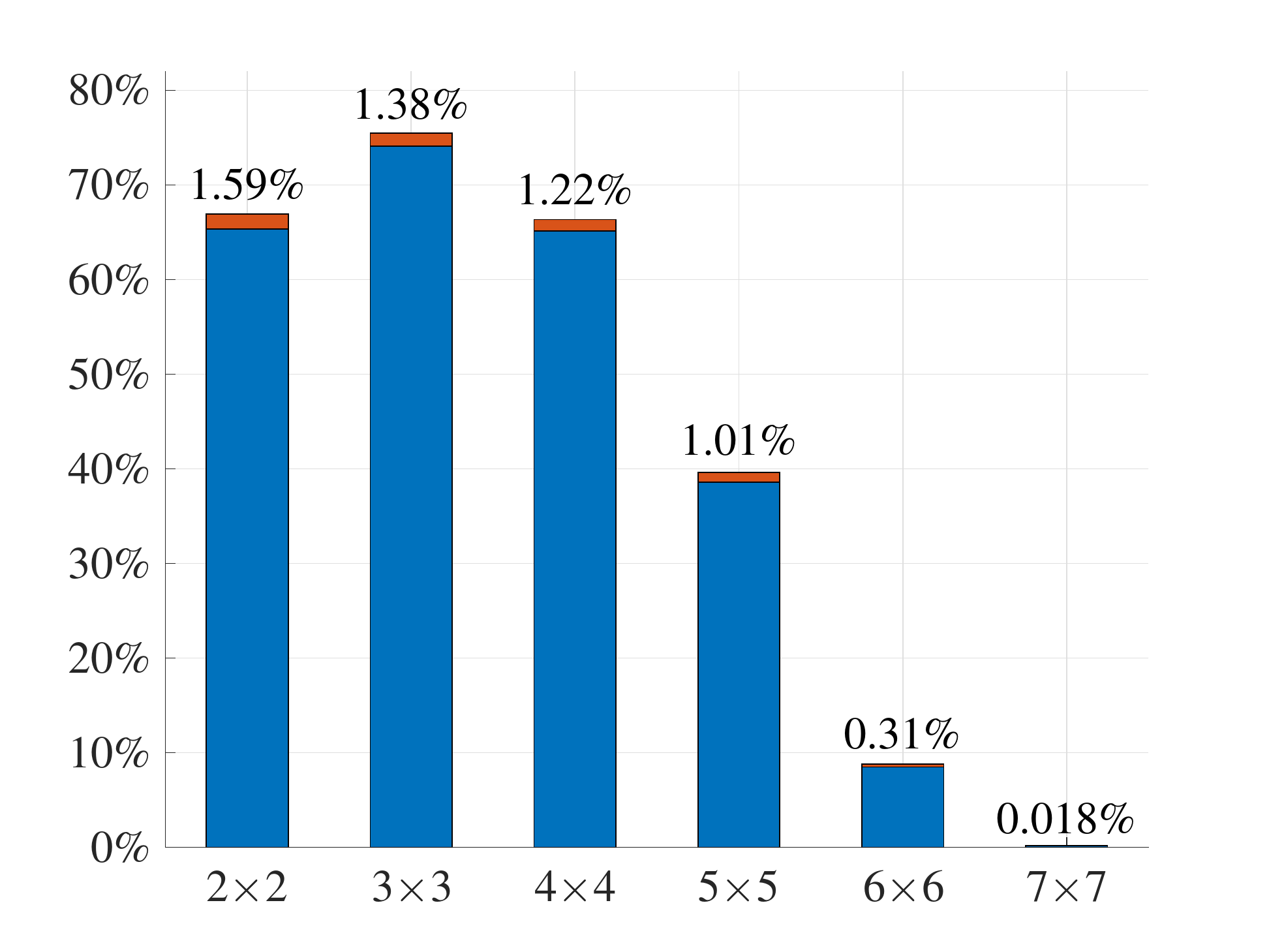}
\caption{\label{fig:higher}
Entanglement detection of $d_A\times d_B$ states with the ESIC and CCNR criteria. The states are
generated randomly according to the Hilbert-Schmidt measure. 
In each case, the ESIC criterion is able to detect more states,  which are represented by the red portion in each bar (value on top).
}
\end{figure}

Last but not least, we show that the ESIC criterion is better than CCNR also  for  higher dimensions. 
However, this conclusion does not mean that the ESIC criterion will get stronger in higher dimensions, as it has been shown already that each single criterion based on positive maps detects smaller fractions of states if the dimension increases; see Refs.~\cite{jmp51.042202, jpa48.505302}.
Altogether $50\,000$ random states are generated according to the Hilbert-Schmidt measure in each dimension ranging from $2\times2$ up to $7\times7$, respectively, see the results in Fig.~\ref{fig:higher}.
In all these cases, the ESIC criterion is shown to be better than CCNR, although the differences vary with the dimension.
It is quite surprising that an unusually large percentage of entangled states is  detected in the $3\times3$ case compared to other dimensions. Dimension 3 has been known to be very special in the study of SIC~POVMs \cite{Zauner}. The reason behind such anomaly deserves further study.

\section{Conclusion}\label{sec:Con}
A  SIC POVM represents a special single measurement setting, which is  tomographically complete.
In this paper we show that by using SIC POVMs we can construct a stronger and more efficient entanglement detection criterion than the CCNR criterion based on local orthogonal observables.  The superiority of our criterion  is illustrated  with various examples.
The reason behind this  superiority is worthy of further study.
In passing, we note that an equivalent criterion can be constructed by replacing  SIC POVMs with the complete sets of mutually unbiased bases.

For the CCNR criterion, the local uncertainty relation provides a straightforward construction of nonlinear entanglement witnesses. Then, an interesting open question is how to achieve a nonlinear improvement for the ESIC criterion with SIC POVMs.
Note that the nonlinear criterion with SIC POVMs introduced recently in Ref.~\cite{qip17.111} is  identical to the criterion using observables in Refs.~\cite{pra78.052319, pra77.060301}.

\acknowledgments
We thank Cheng-Jie Zhang for stimulating discussions.
This work has been supported by the European Research Council (Consolidator Grant No. 683107/TempoQ) and the Deutsche Forschungsgemeinschaft.
J.S. is also supported by the Beijing Institute of Technology Research Fund Program for Young Scholars.
A.A. acknowledges support by Erwin Schr\"{o}dinger Stipendium No. J3653-N27.
H.Z. acknowledges financial support from the Excellence Initiative of the German Federal and State Governments (ZUK 81) and the Deutsche Forschungsgemeinschaft in the early stage of this work.



\begin{thebibliography}{00}
\bibitem{note1}
Random states with the Hilbert-Schmidt distribution can be generated as follows:  first generate a square random matrix $A$ whose entries are drawn independently according to the standard complex Gaussian distribution, then the density matrix $\rho=AA^{\dag}/\tr(AA^{\dag})$ obeys the Hilbert-Schmidt distribution. 
For Fig.~\ref{fig:2qb} in the main text, we only keep entangled states thus generated.
\end{thebibliography}
\end{document}